\documentclass{article}
\usepackage{graphicx} 
\usepackage{multirow}
\usepackage{tabularx}
\usepackage[utf8]{inputenc}
\usepackage{fullpage}
\usepackage{amsmath}
\usepackage{amsthm}
\theoremstyle{plain}
\usepackage{amssymb}
\usepackage[ruled,vlined]{algorithm2e}
\usepackage{xcolor}
\usepackage{hyperref}
\usepackage{mathtools}
\usepackage{mdframed}

\newcommand{\ceil}[1]{\left\lceil #1 \right\rceil}
\newcommand{\ip}[1]{\iprod{#1}}
\newcommand{\iprod}[1]{\langle #1 \rangle}

\newcommand{\set}[1]{\left\{#1\right\}}

\newcommand{\paren}[1]{\left(#1\right)}

\newcommand{\floor}[1]{\left \lfloor#1 \right \rfloor}
\newcommand{\eqdef}{:=}

\newcommand{\N}{{\mathbb{N}}}

\newcommand{\condition}{\;\ifnum\currentgrouptype=16 \middle\fi|\;}

\newcommand{\eps}{\varepsilon}

\newcommand{\la}{\gets}

\makeatletter
\@ifpackageloaded{complexity}{}{%
\usepackage[classfont=roman,langfont=roman,funcfont=roman]{complexity}}
\makeatother

\newcommand{\argmax}{\operatorname*{argmax}}

\newcommand{\MathAlg}[1]{\mathsf{#1}}

\newcommand{\Ensuremath}[1]{\ensuremath{#1}\xspace}

\newcommand{\MathAlgX}[1]{\Ensuremath{\MathAlg{#1}}}

\renewcommand{\Pr}{{\mathrm {Pr}}}

\newcommand{\ppr}[2]{\Pr_{#1}\left[#2\right]}

\newcommand{\Ac}{\MathAlgX{A}}

\newcommand{\size}[1]{\left|#1\right|}

\def\cA{{\cal A}}

\def\cH{{\cal H}}

\def\cX{{\cal X}}
\def\cY{{\cal Y}}

\def\bbN{\N}

\def\bbR{{\mathbb R}}

\newcommand{\Tableofcontents}{
\thispagestyle{empty}
\pagenumbering{gobble}
\clearpage
\tableofcontents
\thispagestyle{empty}
\clearpage
\pagenumbering{arabic}
}

\newcommand{\IPConcave}{\mathsf{IPConcave}}
\newcommand{\IPConcaveHighDim}{\mathsf{IPConcaveHighDim}}

\newcommand{\PrivateIP}{\mathsf{PrivateIP}}

\newcommand{\tx} {\tilde{x}}

\newcommand{\tcX} {\tilde{\cX}}

\newcommand{\Lap}{{\rm Lap}}

\newcommand{\hx}{\hat{x}}

\newtheorem{lemma}{Lemma}
\newtheorem{theorem}{Theorem}
\newtheorem{definition}{Definition}

\newtheorem{fact}{Fact}
\newtheorem{corollary}{Corollary}

\newtheorem{remark}{Remark}

\newcommand{\calX}{{\cal X}}
\newcommand{\calY}{{\cal Y}}
\newcommand{\remove}[1]{}

\title{Differentially Private Quasi-Concave Optimization: Bypassing the Lower Bound and Application to Geometric Problems}

\author{Kobbi Nissim\thanks{Department of Computer Science, Georgetown University. E-mails: {\tt kobbi.nissim@georgetown.edu}, {\tt eliadtsfadia@gmail.com}, {\tt cy399@georgetown.edu}. Work is partially funded by NSF grant No.~2217678 and by a gift to Georgetown University.}
\and
Eliad Tsfadia$^\ast$
\and
Chao Yan$^\ast$}

\begin{document}

\maketitle
\begin{abstract}
    We study the sample complexity of differentially private optimization of quasi-concave functions. For a fixed input domain $\cX$, Cohen et al.~\cite{CohenLNSS22} (STOC 2023) proved that any generic private optimizer for low sensitive quasi-concave functions must have sample complexity $\Omega(2^{\log^*|\cX|})$.
    
    We show that the lower bound can be bypassed for a series of ``natural'' problems. 
    We define a new class of \emph{approximated} quasi-concave functions, and present a generic differentially private optimizer for approximated quasi-concave functions with sample complexity $\tilde{O}\paren{\log^*|\cX|}$. 
    As applications, we use our optimizer to privately select a center point of points in $d$ dimensions and \emph{probably approximately correct} (PAC) learn $d$-dimensional halfspaces. 
    In previous works, Bun et al.~\cite{BNSV15} (FOCS 2015) proved a lower bound of $\Omega(\log^*|\cX|)$ for both problems. 
    Beimel et al.~\cite{beimel2019private} (COLT 2019) and Kaplan et al.~\cite{kaplan2020private} (NeurIPS 2020) gave an upper bound of $\tilde{O}(d^{2.5}\cdot 2^{\log^*|\cX|})$ for the two problems, respectively. We improve the dependency of the upper bounds on the cardinality of the domain by presenting a new upper bound of $\tilde{O}(d^{5.5}\cdot\log^*|\cX|)$ for both problems. 
    To the best of our understanding, this is the first work to reduce the sample complexity dependency on $|\cX|$ for these two problems from exponential in $\log^* |\cX|$ to $\log^* |\cX|$. 
    
\end{abstract}

\Tableofcontents

\section{Introduction}
The training of machine learning models often uses sensitive personal information that requires privacy protection. We explore privacy-preserving techniques in machine learning, a line of research initiated by Kasiviswanathan et al~\cite{KLNRS08}, where a probably approximately correct (PAC) learner is required to preserve differential privacy with respect to its training data. Differential privacy ensures that a privacy attacker would not be able to detect the presence of an input datum.
Formally,
\begin{definition}[Differential Privacy~\cite{DMNS06}]
    Let $\calX$ be a data domain and $\calY$ be an output domain. A (randomized) mechanism $M \colon \cX^n \rightarrow \cY$ is $(\varepsilon,\delta)$-differentially private if for any pair of neighboring datasets $S,S'\in\calX^n$ (i.e., datasets that differ on a single entry), and any event $E\subseteq \calY$, it holds that
    $$
    \Pr[M(S)\in E]\leq e^{\varepsilon}\cdot\Pr[M(S')\in E] +\delta,
    $$
    where the probability is over the randomness of $M$.
\end{definition}

\emph{Sample complexity}, i.e., the required dataset size for a learning task, is one of the main questions in learning theory. 
Without privacy requirements, it is well-known that the sample complexity of PAC learning a concept class $C$ is $\Theta(VC(C))$~\cite{Valiant84}, where $VC(C)\leq \log |C|$ is the Vapnik–Chervonenkis dimension of $C$. 
Kasiviswanathan et al.~\cite{KLNRS08} provide a generic private learner showing that a sample complexity $O(\log|C|)$ suffices for private learning. The characterization of the sample complexity of private learning is still an open problem.

The main motivation of this work is to advance our understanding of the sample complexity of two fundamental geometric problems: privately finding a center point and privately learning halfspaces, both with input points over a finite Euclidean space $\cX^d$. In particular, we focus on how the sample complexity depends on the cardinality of $\cX$. Towards this goal, we revisit a primitive that was introduced in~\cite{BNS13b} and used in prior work on learning halfspaces: optimizing private quasi-concave functions.

\subsection{Existing Results}

Under \emph{pure} differential privacy (i.e., with $\delta=0$), it is known that the sample complexity of privately learning halfspaces with input points over a finite $d$-dimensional domain $\cX^d$ is $\Theta(d^2\log|\cX|)$.\footnote{The upperbound follows from the $O(\log |C|)$ upperbound of~\cite{KLNRS08} (considering that a hyperplne in $d$ dimensions can be represented by $d$ points). The lowerbound follows from~\cite{FX14} who showed that the  Littlestone dimension of halfspaces is $d^2\log|\cX|$.}
Beimel et al.~\cite{BNS13b} showed that the sample complexity of privately learning thresholds (i.e., $1$-dimensional halfspaces) under \emph{approximate} differential privacy (i.e., with $\delta>0$) can be significantly smaller than under pure differential privacy. 
They constructed a general differentially private algorithm $\cA_{RecConcave}$ that, given a quasi-concave low-sensitivity target function $Q \colon \cX^* \times \tcX \rightarrow \bbR$, requires only $\tilde{O}\paren{2^{O(\log^* |\tcX|)}}$ data elements to optimize it. 
I.e., given a dataset $S \in \cX^*$ of that size, $\cA_{RecConcave}$ computes  $\tx \in \tcX$ such that $Q(S,\tx)$ is close to $\max_{x \in \tcX}\set{Q(S,x)}$. The properties that $Q$ should satisfy are:
\begin{enumerate}
    \item \emph{Quasi-Concave}: For any dataset $S \in \cX^*$ and any $x' < x < x''$ in $\tcX$ it holds that $Q(S,x)\geq \min\{Q(S,x'),Q(S,x'')\}$,
    \item \emph{Low-Sensitivity}: For any neighboring $S,S' \in \cX^*$ and any $x\in \tcX$ it holds that $\size{Q(S,x) - Q(S',x)} \leq 1$.
\end{enumerate}

Bun et al.~\cite{BNSV15} and Kaplan et al.~\cite{KaplanLMNS19} reduced learning thresholds to the \emph{interior point} problem, where given a dataset $S \in \cX^n$ for a finite $\cX \subset \bbR$, the goal is to output $x$ such that $\min S \leq x \leq \max S$. Note that the latter problem can be solved privately by optimizing the following quasi-concave function:
\begin{align}\label{def:Q_Thr}
    Q_{IP}(S = \set{x_1,\ldots,x_n}, x) \eqdef \min\set{\size{\set{i \in [n] \colon x_i \leq x}},\size{\set{i \in [n] \colon x_i \geq x}}}.
\end{align}

Beimel et al.~\cite{beimel2019private} and Kaplan et al.~\cite{kaplan2020private} extended this approach to optimize high-dimensional functions. More specifically, suppose that we have a low-sensitivity quasi-concave $d$-dimensional function $Q\colon (\cX^d)^* \times \bbR^d \rightarrow \bbR$ that we would like to optimize.\footnote{A $d$-dimensional $Q$ is quasi-concave if for any dataset $S$, any points $x_1,\ldots,x_k$ and any $x$ in their convex hull, it holds that $Q(S,x) \geq \min_{i}\set{Q(S,x_i)}$.} Then we can apply the $1$-dimensional optimizer $\cA_{RecConcave}$ coordinate-by-coordinate: In step $i \in [d]$, given the results $\tx_1,\ldots,\tx_{i-1}$ of the previous optimizations, we compute $\tx_i$ by applying $\cA_{RecConcave}$ to optimize the $1$-dimensional function
\begin{align}\label{eq:Q_tx}
    Q_{\tx_1,\ldots,\tx_{i-1}}(S,x_i) \eqdef \max_{x_{i+1},\ldots,x_d \in \bbR} Q(S, (\tx_1,\ldots,\tx_{i-1}, x_i,x_{i+1},\ldots,x_d)),
\end{align}
where $x_i$ is chosen from a proper finite domain $\tcX_i$ with $\log^* |\tcX_i| \approx \log^* |\cX|$.

Beimel et al.~\cite{beimel2019private} applied this approach to privately find a center point by privately optimizing the Tukey-Depth function (Definition~\ref{def:TukeyDepth}), and Kaplan et al.~\cite{kaplan2020private} reduced private learning of halfspaces to privately optimizing a quasi-concave function $cdepth$ (Definition~\ref{def:depth-cdepth}). 
The resulting upper bounds for both problems are $\tilde{O}(d^{2.5}\cdot 2^{O(\log^*|\cX|)})$, where the term $2^{O(\log^*|\cX|)}$ in both results is due to the application of the $1$-dimensional quasi-concave optimizer $\cA_{RecConcave}$. 
Towards better understanding the sample complexity of these two problems, the question we ask in this paper is:
\begin{mdframed}
    Are there differentially private algorithms for center points and for PAC learning halfspaces with sample complexity $O(\poly(d)\cdot\log^*|\cX|)$?
\end{mdframed}

A series of works by Bun et al.~\cite{BNSV15}, Beimel et al.~\cite{beimel2019private} and Alon et al.~\cite{AlonLMM19} give a lower bound of $\Omega(d+\log^*|\cX|)$
on the sample complexity of privately selecting a center point and privately learning halfspaces. 
In particular, in terms of the dependency in $\log^*|\cX|$, there is an exponential gap between the upper and lower bounds established prior to this work.

In the $1$-dimensional case, this gap was recently closed. Kaplan et al.~\cite{KaplanLMNS19} presented an upper bound of $\tilde{O}((\log^* |\cX|)^{1.5})$ for learning thresholds, and more recently, Cohen et al.~\cite{CohenLNSS22} improved this upper bound to $\tilde{O}(\log^* |\cX|)$. However, it remained unclear how to extend these upper bounds to the high-dimensional case. The main issue is that these methods are tailored to optimize the specific interior point function $Q_{IP}$ (Equation~\ref{def:Q_Thr}) and do not provide a general method to optimize any quasi-concave function as $\cA_{RecConcave}$ does. Therefore, it was unclear how to apply the ideas there to the more general functions that are induced by optimizing high dimensional tasks coordinate-by-coordinate (Equation~\ref{eq:Q_tx}).

One could hope that a general private quasi-concave optimization can be done using an exponentially smaller sample complexity. Unfortunately, Cohen et al.~\cite{CohenLNSS22} proved that a sample complexity of $\Omega(2^{\log^*|\cX|})$ is necessary, in general. They interpreted this lowerbound as follows:

\medskip
{\em 
``We view this lower bound as having an important conceptual message, because private quasi-concave optimization is the main workhorse (or more precisely, the only known workhorse) for several important tasks, such as privately learning (discrete) halfspaces \cite{beimel2019private,kaplan2020private}. As such, current bounds on the sample complexity of privately learning halfspaces are exponential in $\log^* \size{\cX}$, but it is conceivable that this can be improved to a polynomial or a linear dependency. The lower bound means that either this is not true, or that we need to come up with fundamentally new algorithmic tools in order to make progress w.r.t. halfspace.''~\cite{CohenLNSS22}
}
\medskip


We show how to bypass the lower bound of Cohen et al.~\cite{CohenLNSS22} by only focusing on `natural' quasi-concave functions which include the functions that are induced in Equation~\ref{eq:Q_tx} for the high dimensional optimization tasks of \cite{beimel2019private,kaplan2020private}. As a result, we achieve the first exponential improvement in the dependency on $\log^*|\cX|$ and establish a new upper bound of $\tilde{O}(d^{5.5}\log^*|\cX|)$ for privately learning halfspaces and a selecting center point.\footnote{We are aware that a recent paper by Huang et al.~\cite{ACML-huang25b} claims to privately learn halfspaces using sample complexity of $\tilde{O}(d \cdot \log^* \size{\cX})$. In Appendix~\ref{sec:ACML}, we show that their algorithm is not differentially private.}

\subsection{Our Results}

\subsubsection{Private Optimization for Approximated Quasi-Concave Functions}

Our first contribution is a new method to bypass the lower bound of Cohen et al.~\cite{CohenLNSS22} for a natural class of quasi-concave functions which we call \emph{approximated} quasi-concave functions.

More formally, we say that a function $Q\colon \cX^* \times \tcX \rightarrow \bbR$ can be $(\alpha,\beta,m)$-approximated if a random subset $S'\subseteq S$ of size $|S'|=m$ satisfies $\Pr[\forall x,|Q(S,x)-Q(S',x)|\leq\alpha]\geq 1-\beta$. We consider $\alpha$ as an `acceptable' additive error (in particular, $\alpha \ll \max_{x \in \tcX}\set{Q(S,x)}$), and $\beta$ as a small enough confidence error.

We show how to reduce the task of optimizing such functions $Q$ to the $1$-dimensional interior point problem. The idea is to use the sample and aggregate approach of \cite{NRS07}: We partition the $n$-size input dataset $S$ into random $m$-size subsets, compute (non-privately) the optimal solution with respect to each subset, and then aggregate the $n/m$ solutions using a private interior point algorithm. By the approximation property of $Q$, with probability $1- \beta n /m \approx 1$, all the $n/m$ solutions have high value over $Q(S,\cdot)$ (i.e., at most $\alpha$-far from the optimum), and since the function is quasi-concave, then any interior point also has a high value. By using the $\tilde{O}(\log^* |\tcX|)$ interior-point algorithm of Cohen et al.~\cite{CohenLNSS22}, we obtain the following theorem: 




\begin{theorem}[Informal]\label{def:intro:IPConcave}
    There exists a $(\varepsilon,\delta)$-differentially private algorithm $\IPConcave$ that given a target function $Q \colon \cX^* \times \tcX \rightarrow \bbR$ that is quasi-concave and can be $(\alpha,\beta,m)$-approximated, and given a dataset $S$ of size at least $\tilde{O}_{\eps,\delta}(m\log^*|\tcX|)$, the algorithm outputs $\hat{x}\in \tcX$ such that $$\Pr\left[|Q(S,\hat{x})-\max_{x}Q(S,x)|\leq O(\alpha)\right]\geq 1-\tilde{O}(\beta\log^*|\tcX|).$$
\end{theorem}

Similar to the work of Beimel et al.~\cite{beimel2019private} and Kaplan et al.~\cite{kaplan2020private}, we extend this method to high-dimensional functions by applying $\IPConcave$ iteratively coordinate-by-coordinate.

\begin{theorem}[Informal]\label{def:intro:IPConcaveHighDim}
    There exists a $(\varepsilon,\delta)$-differentially private algorithm $\IPConcaveHighDim$ that given a high-dimensional target function $Q \colon (\cX^d)^* \times \bbR^d \rightarrow \bbR$ that is quasi-concave and can be $(\alpha,\beta,m)$-approximated, and given a dataset $S$ of size at least $\tilde{O}_{\varepsilon,\delta,\alpha,\beta}(m\log^*|\cX| \sqrt{d})$, the algorithm outputs $\hat{x}\in \bbR^d$ such that $$\Pr\left[|Q(S,\hat{x})-\max_{x \in \bbR^d}Q(S,x)|\leq O(\alpha d)\right]\geq 1-\tilde{O}(\beta(\log^*|\cX| + d)).$$
\end{theorem}

Note that the approximation requirement here is a very strong guarantee. It requires that a random subset has a value similar to this of the original dataset for `every' input $x$. Fortunately, VC theory can provide such a guarantee for the functions that were used by \cite{beimel2019private,kaplan2020private}. As applications, we close the exponential gap of sample complexity for private center point and learning halfspaces, as described next.

\subsubsection{Applications}

\paragraph{Private Center Point}
For privately approximating the center point, we follow the approach of Beimel et al.~\cite{beimel2019private} to optimize the Tukey depth function, but now with our new $\IPConcaveHighDim$ method. By VC theory, the Tukey depth function (normalized by the dataset size) can be $(\alpha,\beta,m)$-approximated with $m = \tilde{O}((d + \log(1/\beta))/\alpha^2)$ (Lemma~\ref{lem:VC theory 1}). Furthermore, the Tukey center of every $n$-size dataset $S$ has Tukey depth at least $\frac{n}{d+1}$. Hence, by substituting $\alpha$ with $\alpha/d^2$ in Theorem~\ref{def:intro:IPConcaveHighDim}, we obtain the following result.

\begin{theorem}[Privately approximating the center point]\label{thm:intro:TD}
Let $\cX \subset \bbR$ be a finite domain.
There exists an $(\varepsilon,\delta)$-differentially private algorithm $\Ac \colon (\cX^d)^n \rightarrow \bbR^d$ that given an $n$-size dataset $S \in (\cX^d)^n$, for $n \geq \tilde{\Theta}_{\alpha,\beta,\varepsilon,\delta}\left(d^{5.5}\cdot\log^*|\calX|\right)$, $\Ac(S)$ outputs, with probability $1-\beta$, a point with Tukey depth in $S$ of at least $\frac{1 - \alpha}{d+1} \cdot n$.
\end{theorem}

\paragraph{Private Halfspace Learning}

For privately learning halfspaces, we follow the approach of Kaplan et al.~\cite{kaplan2020private} to reduce the task to private feasibility problem and solve this problem by optimizing a quasi-concave function  $cdepth$ (Definition~\ref{def:depth-cdepth}), but now using our $\IPConcaveHighDim$ method. Similarly to the Tukey depth function, the $cdepth$ function (normalized by the dataset size) can also be $(\alpha,\beta,m)$-approximated for $m = \tilde{O}((d + \log(1/\beta))/\alpha^2)$ (Corollary~\ref{corollary:vc for cdepth}). 
An $\alpha/d$ error in the $cdepth$ function is translated to an $\alpha$ error for learning halfspace and the linear feasibility problem. Therefore, we achieve our learner by substituting $\alpha$ with $\alpha/d^2$ in Theorem~\ref{def:intro:IPConcaveHighDim}.

\begin{theorem}[Privately learning halfspaces]\label{thm:intro:halfspaces}
    Let $\cX \subset \bbR$ be a finite domain. 
    There exists an $(\eps,\delta)$-differentially private $(\alpha,\beta)$-PAC learner for halfspaces over examples from $\cX^d$ with sample complexity $n = \tilde{O}_{\alpha,\beta,\varepsilon,\delta}\left(d^{5.5}\cdot\log^*|\cX|\right)$.
\end{theorem}

\paragraph{Learning over concept class of VC dimension $1$.}
We remark that since every concept class with VC dimension 1 can be embedded in 3-dimensional halfspaces class \cite{alon2017sign,AlonHW87}, our results imply that every concept class with VC dimension 1 can be privately learned with sample size $\tilde{O}(\log^*|\cX|)$.

\begin{corollary}
    For any concept class $\mathcal{C}$ with VC dimension 1 and input domain $\cX$, there exists an $(\varepsilon,\delta)$-differentially private algorithm that can $(\alpha,\beta)$-PAC learn it with sample size $\tilde{O}_{\varepsilon,\delta,\alpha,\beta}(\log^*|\cX|)$.
\end{corollary}

\subsection{Open Questions}

In this work, we present a general private quasi-concave optimization method that bypasses the lower bound of Cohen et al.~\cite{CohenLNSS22} and enables to achieve the first $d^{O(1)} \cdot\log^* |\cX|$ sample complexities for fundamental problems such as privately approximating the center point and privately learning halfspaces over a finite euclidean input domain $\cX^d$. But while our result bridge the exponential gap in $\log^* |\cX|$, the dependency on the dimension $d$ has increased from $d^{2.5}$ (\cite{beimel2019private,kaplan2020private}) to $d^{5.5}$ (Theorems~\ref{thm:intro:TD} and \ref{thm:intro:halfspaces}). It remains open to understand if it is possible to reduce the dependency on $d$ while avoiding an exponential blow-up in $\log^* |\cX|$.


\section{Preliminaries}

\subsection{Notations}

We denote a subset $S$ of $X$ by  $S \subseteq X$, and a multi-set $S$ of elements in $X$ by $S \in X^*$ (or $S \in X^n$ if $S$ is of size $n$). For a multi-set $S \in X^*$ and a set $T \subseteq X$, we let $S \cap T$ be the multi-set of all the elements in $S$ (including repetitions) that belong to $T$. A subset $S'$ of a multi-set $S$ (denoted by $S' \subseteq S$) refers to a sub-multiset, i.e., if $S = \set{x_1,\ldots,x_n}$ then there exists a set of indices $I = \set{i_1,\ldots,i_m} \subseteq [n]$ such that $S' = \set{x_{i_1},\ldots,x_{i_m}}$. A random $m$-size subset $S'$ of a multi-set $S$ is specified by a random $m$-size set of indices $I \subseteq [n]$.

Throughout this paper, a dataset refers to a multi-set.

\subsection{Learning Theory}
\begin{definition}[Vapnik-Chervonenkis dimension~\cite{VC,Haussler1986EpsilonnetsAS}]
    Let $X$ be a set and $R$ be a set of subsets of $X$. Let $S\subseteq X$ be a subset of $X$. Define
    $
    \Pi_R(S)=\{S\cap r \mid r\in R\}.
    $
    If $|\Pi_R(S)|=2^{|S|}$, then we say $S$ is shattered by $R$. The Vapnik-Chervonenkis dimension of $(X,R)$ is the largest integer $d$ such that there exists a subset $S$ of $X$ with size $d$ that is shattered by $R$.
\end{definition}


\begin{definition}[$\alpha$-approximation\footnote{Commonly called $\varepsilon$-approximation. We use $\alpha$-approximation as $\varepsilon$ is used as a parameter of differential privacy.}~\cite{VC,Haussler1986EpsilonnetsAS}]\label{def:alpha-approx}
    Let $X$ be a set and $R$ be a set of subsets of $X$. 
    Let $S \in X^*$ be a finite multi-set of elements in $X$.
    For any $0\leq\alpha\leq 1$ and $S'\subseteq S$, $S'$ is an $\alpha$-approximation of $S$ for $R$ if for all $r\in R$, it holds that
    $
    \left|\frac{|S\cap r|}{|S|}-\frac{|S'\cap r|}{|S'|}\right|\leq \alpha.
    $
\end{definition}
    
\begin{theorem}[\cite{VC,Haussler1986EpsilonnetsAS}]\label{thm:alphaApprox}
    Let $(X,R)$ have VC dimension $d$. Let $S\subseteq X$ be a subset of $X$. Let $0<\alpha,\beta\leq 1$. Let $S'\subseteq S$ be a random subset of $S$ with size at least 
    $ O\left(\frac{d\cdot\log\frac{d}{\alpha}+\log\frac{1}{\beta}}{\alpha^2}\right).
    $
    Then with probability at least $1-\beta$, $S'$ is an $\alpha$-approximation of $S$ for $R$.\footnote{Although the theorem is stated for sets $S \subseteq X$, it also holds for multisets $S \in X^*$. To see it, fix an $n$-size multiset $S \in X^n$ and transform the domain set $X=\{x_1,\dots,x_k\}$ to an extended domain set $X'=\{x_{1,1},\dots,x_{1,n},\dots,x_{k,1},\dots,x_{k,n}\}$. For every $r\in R$, we can transform it to $r'=\{x_{i,1},\dots, x_{i,n}: x_i\in r\}$, and let $R'=\{r':r\in R\}$. We have $VC(X',R')=VC(X,R)$. Now each element $x_i$ that appears $t$ times in $S$ can be replaced by $x_{i,1},x_{i,2},\dots, x_{i,t}$. This reduction transforms all multiset setting to a corresponding set setting.}
 \end{theorem}

\begin{definition}[Generalization and empirical error]
    Let $\mathcal{D}$ be a distribution, $c$ be a concept and $h$ be a hypothesis. The error of $h$ w.r.t.~$c$ over $\mathcal{D}$ is defined as 
    $$
    error_{\mathcal{D}}(c,h)=\Pr_{x\sim\mathcal{D}}[c(x)\neq h(x)].
    $$
    For a finite dataset $S$, the error of $h$ w.r.t.~$c$ the over $S$ is defined as 
    $$
    error_S(c,h)=\frac{\left|\set{x\in S \mid c(x)\neq h(x)}\right|}{|S|}.
    $$
\end{definition}

\begin{theorem}[\cite{BlumerEhHaWa89,kaplan2020private}]\label{thm:learn vc}
    Let $\mathcal{C}$ be a concept class and let $\mathcal{D}$ be a distribution. Let $\alpha,\beta>0$, and $m\geq \frac{48}{\alpha}\left(10VC(\mathcal{C})\log(\frac{48e}{\alpha})+\log(\frac{5}{\beta})\right)$. Let $S$ be a sample of $m$ points drawn i.i.d.\ from $\mathcal{D}$. 
    Then
    $$\Pr[\exists c,h\in\mathcal{C}~\mbox{s.t.}~error_{S}(c,h)\leq\alpha/10 ~\mbox{and}~error_{\mathcal{D}}(c,h)\geq\alpha]\leq \beta.$$
\end{theorem}

\subsection{Differential Privacy}

\begin{definition}[Differential Privacy~\cite{DMNS06}]
    Let $\calX$ be a data domain and $\calY$ be an output domain. A (randomized) mechanism $M$ mapping $\calX^n$ to $\calY$ is $(\varepsilon,\delta)$-differentially private if for any pair of inputs $S,S'\in\calX^n$  where $S$ and $S'$ differ on a single entry, and any event $E\subseteq \calY$, it holds that
    $$
    \Pr[M(S)\in E]\leq e^{\varepsilon}\cdot\Pr[M(S')\in E] +\delta,
    $$
    where the probability is over the randomness of $M$.
\end{definition}

\begin{theorem}[Advanced composition~\cite{DRV10}]\label{thm:advancedcomposition}
    Let $M_1,\dots,M_k:\calX^n\rightarrow \calY$ be $(\varepsilon,\delta)$-differentially private mechanisms. Then the algorithm that on input $S\in \calX^n$ outputs $(M_1(S),\dots,M_k(S))$ is $(\varepsilon',k\delta+\delta')$-differentially private, where $\varepsilon'=\sqrt{2k\ln(1/\delta')}\cdot \varepsilon$ for every $\delta'>0$.
\end{theorem}

We use differentially private algorithms for the Interior Point problem:

\begin{definition}[Interior Point~\cite{BNSV15}] Let $\calX$ be a (finite) ordered domain.
    We say that $p\in \calX$ is an interior point of a dataset $S=\{x_1,\dots,x_n\} \in \cX^n$ if $\min_{i\in[n]} \set{x_i}\leq p\leq \max_{i\in[n]} \set{x_i}$.
\end{definition}

Bun et.~al.~\cite{BNSV15} proved that the sample complexity for solving the interior point problem with  differential privacy must grow proportionally to $\log^*|\calX|$~\cite{BNSV15}. 
Cohen et.~al.\ provide a nearly optimal algorithm solving the interior point privately on a finite domain~\cite{CohenLNSS22}.
\begin{theorem}[\cite{CohenLNSS22}]\label{thm:1-d interior point}
    Let $\calX$ be a finite ordered domain. 
    There exists an $(\varepsilon,\delta)$-differentially private algorithm $\PrivateIP$ that on input $S\in \cX^n$ outputs an interior point with probability $1-\beta$ provided that $n > n_{IP}(|\cX|,\beta,\eps,\delta)$ for $n_{IP}(|\cX|,\beta,\eps,\delta) \in O\left(\frac{\log^*|\calX|\cdot\log^2(\frac{\log^*|\calX|}{\beta\delta})}{\varepsilon}\right)
    $.
\end{theorem}

\subsection{Halfspaces}
\begin{definition}[Halfspaces and hyperplanes]
    Let $\cX\subset \mathbb{R}^d$. For $a_1,\dots,a_d,w\in\mathbb{R}$, the halfspace  predicate $h_{a_1,\dots,a_d,w}:\cX\rightarrow\{\pm1\}$ is defined as $h_{a_1,\dots,a_d,w}(x_1,\dots,x_d)=1$ if and only if $\sum_{i=1}^d a_ix_i\geq w$. Define the concept class $\mbox{HALFSPACE}_d(\cX)=\{h_{a_1,\dots,a_d,w}\}_{a_1,\dots,a_d,w\in\mathbb{R}}$. 
\end{definition}

\subsection{Notation}

We use the notation $(x_1,\ldots,x_i)\times \calX^{d-i}$ for the space of points with a fixed prefix of $i$ coordinates:
$$
(x_1,\ldots,x_i)\times \calX^{d-i}=
\left\{y\in\calX^d~|~y_j=x_j~\mbox{for }~j\in[i]\right\}.
$$

\section{Our Private Quasi-Concave Optimization Scheme}

In this section, we present our main algorithm $\IPConcave$ that for privately optimizing approximated quasi-concave functions.

We first give the definition of the quasi-concave function and the sensitivity of functions.

\begin{definition}[Quasi-Concave]
    Let $\cX$ be an ordered domain.
    A function $f\colon \cX \rightarrow \bbR$ is quasi-concave if $f(\ell)\geq \min(\{f(i),f(j)\}$ for every $i<\ell<j$.
\end{definition}

\begin{definition}[Sensitivity]
    The sensitivity of a function $f \colon \cX^* \rightarrow \bbR$ is the smallest $k$ such that for every pair of neighboring datasets $S,S' \in \cX^*$ (i.e., differ in exactly one entry), we have $\size{f(S)-f(S')} \leq k$. A function $Q \colon \cX^* \times \tcX \rightarrow \bbR$ is called a sensitivity-$k$ function if for every $x \in \tcX$, the function $Q(\cdot, x)$ has sensitivity $\leq k$.
\end{definition}

\subsection{One-Dimensional Case}

Beimel et al.~\cite{BNS13b} provide an algorithm $\mathcal{A}_{RecConcave}$ that given as inputs a dataset $S \in \cX^*$ and a sensitivity-$1$ function $Q \colon \cX^* \times \tcX \rightarrow \bbR$ such that $Q(S,\cdot)$ is quasi-concave, the algorithm privately finds a point $\hat{x} \in \tcX$ such that $|Q(S,\hat{x})-\max_{x\in \tcX}\{Q(S,x)\}|\leq\tilde{O}\left(2^{O(\log^*|\tcX|)}\right)$. In fact, the exponential dependency in $2^{O(\log^*|\tcX|)}$ is necessary in general since Cohen et al.~\cite{CohenLNSS22} proved a matching lower bound.

In this work, we bypass the lower bound of \cite{CohenLNSS22} by showing that if the quasi-concave function $Q \colon \cX^* \times \tcX \rightarrow \bbR$ has the property that given a dataset $S \in \cX^*$, the function $Q(S,\cdot)$ can be well approximated by $Q(S',\cdot)$ for an $m$-size random subset $S' \subset S$, then we can privately optimize it using sample complexity $\tilde{O}(m\cdot \log^* |\tcX|)$.

Here we define $(\alpha,\beta,m)$-approximation with respect to $Q$.

\begin{definition}[Approximation with respect to $Q$]~\label{def:approximation-wrt-Q}
    For any dataset $S\in \cX^*$ and function $Q \colon \cX^* \times \tcX \rightarrow \bbR$, we say that $S'\subseteq S$ is an $\alpha$-approximation of $S$ with respect to $Q$ if for any $x\in \tcX$, we have $|Q(S,x)-Q(S',x)|\leq \alpha$.
    We say $(S,Q)$ can be $(\alpha,\beta, m)$-approximated, if by randomly selecting a subset $S'\subseteq S$ of size at least $m$, then with probability $1-\beta$, $S'$ is an $\alpha$-approximation of $S$ with respect to $Q$.
\end{definition}

Note that $\alpha$-approximation with respect to a function $Q$ (Definition~\ref{def:approximation-wrt-Q}) is similar to $\alpha$-approximation with respect to a set of binary value functions $R$ (Definition~\ref{def:alpha-approx}). Indeed, in Sections~\ref{sec:Tukey} and \ref{sec:halfspaces} we exploit this similarity and apply Theorem~\ref{thm:alphaApprox} for upper bounding the subset size of approximating specific functions.

In Appendix~\ref{sec:approx-is-low-sen}, we show that approximated function are, in particular, low-sensitivity functions.

In the following, we present our new private optimization algorithm for quasi-concave functions that can be approximated by a random subset.

\begin{algorithm}\label{alg:IPConcave}
    \caption{$\IPConcave$}

    \textbf{Parameter:} Confidence parameter $\beta > 0$, privacy parameter $\eps,\delta > 0$, and number of subsets $t$. 

    \textbf{Inputs:} A dataset $S \in \cX^n$, for $n \geq t$, and a function $Q \colon \cX^* \times \tcX \rightarrow \bbR$ where $\tcX \subseteq \bbR$ and is finite.

    \textbf{Operation:}~

    \begin{enumerate}
        \item Randomly partition $S$ into $S_1,\dots,S_{t}$, each with size at least $\floor{n/t}$.\label{step:partition}

        \item For $i\in [t]$: Compute $y_i = \argmax_{x\in \tcX} \{Q(S_i,x)\}$.

        \item Compute and output $\hat{x} \sim \PrivateIP_{\tcX,\beta,\eps,\delta}(y_1,\dots,y_t)$ (the algorithm from Theorem~\ref{thm:1-d interior point} that solves the 1-dimensional interior point problem over the domain $\tcX$ with confidence parameter $\beta$ and privacy parameters $\eps,\delta$).        

    \end{enumerate}    
\end{algorithm}

\begin{theorem}[Restatement of Theorem~\ref{def:intro:IPConcave}]\label{thm:IPConcave}
    Let $\eps > 0$, $\delta, \alpha,\beta \in (0,1)$, $n,t \in \bbN$ with $n\geq t$, let $\cX$ be a domain of data elements, let $\tcX \subseteq \bbR$ be a finite domain, and let $Q \colon \cX^* \times \tcX \rightarrow \bbR$. Then the following holds:

    \begin{enumerate}
        \item \textbf{Privacy}: Algorithm $\IPConcave_{\alpha,\beta,\eps,\delta,t}(\cdot, Q)\colon \cX^n \rightarrow \tcX$ is $(\eps,\delta)$-differentially private.
        \item \textbf{Accuracy}: Let $S \in \cX^n$ and assume that $Q(S,\cdot)$ is quasi-concave and that $(S,Q)$ can be $(\alpha,\beta',\floor{n/t})$-approximated (Definition~\ref{def:approximation-wrt-Q}). If $t \geq n_{IP}(|\tcX|,\beta,\eps,\delta) \in \tilde{O}_{\beta,\eps,\delta}(\log^* |\tcX|)$ (the sample complexity of $\PrivateIP$ from Theorem~\ref{thm:1-d interior point}), then
        \begin{align*}
            \Pr_{\hat{x} \sim \IPConcave_{\alpha,\beta,\eps,\delta,t}(S,Q)}\left[|Q(S,\hat{x})-\max_{x\in \tcX}\{Q(S,x)\}|\leq 2\alpha\right]\geq 1-t\cdot\beta'-\beta
        \end{align*}
    \end{enumerate}
\end{theorem}
\begin{proof}
    \textbf{Privacy}: A change of one input point affect one of the points $y_i$. So $\PrivateIP$ guarantees that the output is $(\eps,\delta)$-differentially private.

    \textbf{Accuracy}: The $(\alpha,\beta',\ceil{n/t})$-approximation guarantees that for any subset $S_i$, with probability $1-\beta'$, $S_i$ is an $\alpha$-approximation of $S$ with respect to $Q$ (Definition~\ref{def:approximation-wrt-Q}). Let $x_{opt}=\argmax_{x\in \tcX} \{Q(S,x)\}$. Under the condition of $\alpha$-approximation with respect to $Q$, we have $Q(S_i,x_{opt})\geq Q(S,x_{opt})-\alpha=\max_{x\in\mathcal{X}} \{Q(S,x)\}-\alpha$. Then we have $Q(S_i,y_i)\geq Q(S_i,x_{opt})\geq\max_{x\in \tcX} \{Q(S,x)\}-\alpha$. Using the $\alpha$-approximation w.r.t $Q$ again, we have $Q(S,y_i)\geq Q(S_i,y_i)-\alpha\geq \max_{x\in \tcX} \{Q(S,x)\}-2\alpha$. Since $\PrivateIP$ succeeds with probability $1-\beta$, then by the union bound, with probability $1 - t\cdot \beta' - \beta$, for all $i \in [t]$ we have $Q(S,y_i)\geq \max_{x\in \tcX} \{Q(S,x)\}-2\alpha$ and that $\hat{x}$ is an interior point of $\set{y_1,\ldots,y_t}$. Since $Q$ it quasi-concave, the above implies that $Q(S,\hat{x}) \geq \min_{i \in [t]}\set{Q(S,y_i)} \geq \max_{x\in \tcX} \{Q(S,x)\}-2\alpha$, as required.
\end{proof}

\subsection{Extending to Higher Dimension}

Beimel et al.~\cite{beimel2019private} and Kaplan et al.~\cite{kaplan2020private} use $\mathcal{A}_{RecConcave}$ to optimize high dimension functions. The method is to iteratively select good coordinates using $\mathcal{A}_{RecConcave}$. Thus the sample size of their result must depend on $2^{\log^*|\calX|}$. This work shows that our new optimizer can also be extended to higher dimensions if the target function $Q$ is (high-dimensional) quasi-concave and has proper finite-size domains. Notice that the 1-dimensional concavity property is equivalent to quasi-concave.

\begin{definition}[(High-Dimensional) Quasi-Concave]~\label{def:concave-property}
    We say a function $f:\mathbb{R}^d\rightarrow \mathbb{R}$ is \emph{quasi-concave} if for any points $p_1,\dots,p_k$ and any point $p$ in the convex hull of $\{p_i\}_{i\in[k]}$, it holds that $f(p)\geq \min(\{f(p_1),\dots,f(p_k)\})$.
\end{definition}

Since the interior point can only be privately found in a finite domain, we need to construct such domain that contains a point with a high value of $Q$.

\begin{definition}[Proper Finite Domains]~\label{def:proper-finite-domains}
    We say that $\tcX_1,\dots,\tcX_d$ are proper finite domains for $Q \colon\cX^* \times \bbR^d \rightarrow \bbR$ if for any $i\in[d]$, $S \in \cX^*$ and $\hat{x}_1\in\tilde{\calX}_1,\dots,\hat{x}_{i-1}\in\tilde{\calX}_{i-1}$, there exists $\hat{x}_i\in\tilde{\calX}_i$, such that
    $$
    \max_{x_i,\dots,x_d\in\mathbb{R}} Q(S,(\hat{x}_1,\dots,\hat{x}_{i-1},x_i,\dots,x_d))=\max_{x_{i+1},\dots,x_d\in\mathbb{R}}Q(S,(\hat{x}_1,\dots,\hat{x}_{i},x_{i+1},\dots,x_d)).
    $$
    The construction of $\tilde{\cX}_i$ may depend on $\hat{x}_1,\dots,\hat{x}_{i-1}$, so $\tcX_2,\ldots,\tcX_d$ are treated as functions.
\end{definition}

We next describe our algorithm $\IPConcaveHighDim$ that extends $\IPConcave$ for high-dimensional quasi-concave functions. 

\begin{algorithm}~\label{alg:IPConcaveHighDim}

    \textbf{Parameter:} Utility parameters $\alpha,\beta > 0$, privacy parameter $\eps,\delta > 0$, and number of subsets $t$. 

    \textbf{Inputs:} A dataset $S \in \cX^n$, for $n \geq t$, and a function $Q \colon \cX^* \times \bbR^d \rightarrow \bbR$ with proper finite domains $\tcX_1, \tcX_2 \ldots, \tcX_d$ (Definition~\ref{def:proper-finite-domains}).
    
    \caption{$\IPConcaveHighDim$}

    \textbf{Operation:}~
    \begin{enumerate}
    
        \item Randomly partition $S$ into $S_{1},\dots,S_{t}$, each with size at least $\floor{n/t}$.
    
        \item For $i=1,\dots,d$:
        \begin{enumerate}
             
            \item Let $\tilde{\calX}_i = \tilde{\calX}_i(\hat{x}_1,\dots,\hat{x}_{i-1})$, and define the function $\hat{Q}_{\hat{x}_1,\dots,\hat{x}_{i-1}} \colon \cX^n \times \tcX_i \rightarrow \bbR$ as \label{step:hatQ}
            \begin{align*}
                \hat{Q}_{\hat{x}_1,\dots,\hat{x}_{i-1}}(S,x) \eqdef \max_{x_{i+1}\ldots,x_d\in \mathbb{R}} Q(S,(\hat{x}_1,\dots,\hat{x}_{i-1}, x, x_{i+1}\ldots,x_d)).
            \end{align*}

            \item Compute $\hat{x}_i \sim \IPConcave_{\alpha,\beta,\eps,\delta,t}(S,\hat{Q}_{\hat{x}_1,\dots,\hat{x}_{i-1}})$, where we fix the partition in Step~\ref{step:partition} of $\IPConcave$ (Algorithm~\ref{alg:IPConcave}) to $S_{1},\dots,S_{t}$.\label{step:IPConcave} 

        \end{enumerate}

        \item Output $(\hat{x}_1,\dots,\hat{x}_d)$.
    \end{enumerate}
        
\end{algorithm}

\begin{remark}
    Note that in Algorithm~\ref{alg:IPConcaveHighDim}, we chose to use the same partition $S_1,\ldots,S_t$ for all iterations rather than letting $\IPConcave$ to sample a fresh partition in each invocation. The main reason we chose to do that is to increased the confidence guarantee, since once $S_1,\ldots,S_t$ are all a good approximation of $S$ w.r.t. $Q$, they are good for all the iterations and therefore, there is no need to re-sample. 
\end{remark}

\begin{theorem}[Restatement of Theorem~\ref{def:intro:IPConcaveHighDim}]\label{thm:IPConcaveHighDim}
    Let $\eps > 0$, $\delta, \alpha,\beta \in (0,1)$, $n,t,d \in \bbN$ with $n \geq t$, let $\cX$ be a domain of data elements, let $Q \colon \cX^* \times \bbR^d \rightarrow \bbR$ be a function with finite domains $\tcX_1, \tcX_2, \ldots, \tcX_d$, let $X = \max_{i \in [n]}\set{|\tcX_i|}$. Then the following holds:

    \begin{enumerate}
        \item \textbf{Privacy}: Algorithm $\IPConcaveHighDim_{\alpha,\beta,\eps,\delta,t}(\cdot, Q) \colon \cX^n \rightarrow \bbR^d$ is $(\eps \cdot\sqrt{2d\ln(1/\delta')},d\delta+\delta')$-differentially private for any choice of $\delta'>0$.
        \item \textbf{Accuracy}: 
        Let $S \in \cX^n$ and assume that: (1) $Q(S,\cdot)$ is quasi-concave (Definition~\ref{def:concave-property}), (2) $(S,Q)$ can be $(\alpha,\beta',\floor{n/t})$-approximated (Definition~\ref{def:approximation-wrt-Q}), and (3) $\tcX_1, \ldots \tcX_d$ are proper for $Q$ (Definition~\ref{def:proper-finite-domains}). Then for $t = n_{IP}(X,\beta,\eps,\delta) \in \tilde{O}_{\alpha,\beta,\eps,\delta}(\log^* X)$ (the sample complexity of $\PrivateIP$ from Theorem~\ref{thm:1-d interior point} over domain of size $X$), it holds that
        \begin{align*}
            \Pr_{\hat{x} \sim \IPConcaveHighDim_{\alpha,\beta,\eps,\delta,t}(S,Q)}\left[|Q(S,\hat{x})-\max_{x\in \bbR^d}\{Q(S,x)\}|\leq 2\alpha d\right]\geq 1-t\cdot\beta'-d \cdot \beta
        \end{align*}
    \end{enumerate}
\end{theorem}

\begin{proof}
    \noindent {\bf Privacy of Algorithm~\ref{alg:IPConcaveHighDim}:} 
    Each call to $\IPConcave_{\alpha,\beta,\eps,\delta,t}$ (Step~\ref{step:IPConcave}) is $(\epsilon,\delta)$-differentially private (Theorem~\ref{thm:IPConcave}). Using advanced composition (Theorem~\ref{thm:advancedcomposition}), we get that Algorithm~\ref{alg:IPConcaveHighDim} is $(\varepsilon\cdot\sqrt{2d\ln(1/\delta')},d\delta+\delta')$-differentially private for any choice of $\delta'>0$. 
    
    \smallskip
    
    \noindent {\bf Accuracy of Algorithm~\ref{alg:IPConcaveHighDim}}
    By Definition~\ref{def:approximation-wrt-Q} and the union bound, with probability $1-t\cdot \beta'$, all $S_1,\dots,S_{t}$ are $\alpha$-approximation of $S$ with respect to $Q$, and in the following we assume that this event occurs. This means $|Q(S,x)-Q(S_j,x)|\leq \alpha$ for all $x\in\mathbb{R}^d$ and $j\in[t]$. In particular, this implies that for every $i \in [d]$, in the $i$'th iteration of $\IPConcaveHighDim$, the subset $S_j$, for every $j \in [t]$, is an $\alpha$-approximation of $S$ with respect to the (one-dimensional) function $\hat{Q}_{\hat{x}_1,\ldots,\hat{x}_{i-1}}$ (defined in Step~\ref{step:hatQ}). Furthermore, since $Q(S,\cdot)$ is quasi-concave (Definition~\ref{def:concave-property}), then each function $\hat{Q}_{\hat{x}_1,\ldots,\hat{x}_{i-1}}(S,\cdot)$ is also quasi-concave.   
    Thus, by the accuracy guarantee of $\IPConcave$ (Theorem~\ref{thm:IPConcave}), in the $i$'th iteration of $\IPConcaveHighDim$, with probability $1-\beta$, the algorithm computes $\hat{x}_i$ such that
    \begin{align}\label{eq:hatQ}
        \hat{Q}_{\hx_1,\ldots,\hx_{i-1}}(S,\hat{x}_i)
        &\geq \max_{x_i \in \cX_i, x_{i+1}, \ldots, x_d \in \bbR} Q(S,(\hx_1,\ldots,\hx_{i-1},x_i,\ldots,x_d))-2\alpha\\
        &= \max_{x_i, \ldots, x_d \in \bbR} Q(S,(\hx_1,\ldots,\hx_{i-1},x_i,\ldots,x_d))-2\alpha,\nonumber
    \end{align}
    where the equality holds since $\cX_1,\ldots,\cX_d$ are proper finite domains for $Q$ by assumption. 
    Thus, with probability $1 - d\cdot \beta$, Equation~\ref{eq:hatQ} holds for any iteration $i$ during the execution of $\IPConcaveHighDim$, which means that 
    $\hat{Q}(\hx_1) \geq \max_{x \in \bbR^d} Q(S,x)-2\alpha$, and for every $i \in \set{2,\ldots,d}$: 
    \begin{align*}
        \hat{Q}_{\hx_1,\ldots,\hx_{i-1}}(S,\hx_i) \geq \max_{x_i, \ldots, x_d \in \bbR} Q(S,(\hx_1,\ldots,\hx_{i-1},x_i,\ldots,x_d))-2\alpha
        &= \hat{Q}_{\hx_1,\ldots,\hx_{i-2}}(S,\hx_{i-1}) - 2\alpha,
    \end{align*}
    and in the following, we assume that this event occurs.
    We conclude that the output $\hx = (\hx_1,\ldots,\hx_d)$ satisfies $Q(S,\hx) \geq \max_{x \in \bbR^d} Q(S,x) - 2\alpha d$, as required.
    
\end{proof}

\section{Differentially Private Tukey Median Approximation}\label{sec:Tukey}

In this section, we show how to use $\IPConcaveHighDim$ (Algorithm~\ref{alg:IPConcaveHighDim}) to privately approximating the Tukey median. We give additional preliminaries in Section~\ref{sec:pre-tukey-median}, and describe our Tukey median approximation algorithm in Section~\ref{sec:general-tukey-median}.

\subsection{Additional Preliminaries for Tukey Median}~\label{sec:pre-tukey-median}

\begin{definition}[Tukey depth, Tukey median~\cite{Tukey1975MathematicsAT}]\label{def:TukeyDepth}
    Let $S\in (\mathbb{R}^d)^n$ be a dataset of $n$ points and let $p\in\mathbb{R}^d$ be a point (not necessarily in $S$). The {\em Tukey depth} of $p$ in $S$ is the minimum over all hyperplanes $h$ going through $p$ of the number of points in $S$ on one side of $h$. A point is a {\em Tukey median} of $S$ if it has the maximum Tukey depth in $S$.
\end{definition}

We use $TD_S(p)$ for the Tukey depth of point $p$ in point set $S$. 

\begin{fact}~\label{fact:tukey depth of point in convex hull} 
    Let $S\in (\mathbb{R}^d)^n$ and let $p_1,\dots,p_k$ be points in $\mathbb{R}^d$ with Tukey depth at least $\gamma n$ in $S$. Then any point in the convex hull of $p_1,\dots,p_k$ has Tukey depth at least $\gamma n$ in $S$.
\end{fact}
\begin{proof}
    Let $p'$ be a point inside the convex hull of $p_1,\dots,p_k$ and assume $p'$ has Tukey depth $T<\gamma n$. Then there exists a hyperplane $h$ that goes through $p'$ with less than $\gamma n$ points on one of its sides. Denote by $S_1$ and $S_2$ be two half-spaces on the two sides of $h$. Without loss of generality, $S_1$ contains $T$ points of $S$. Note that since  $p'$ is in the convex hull of $p_1,\dots,p_k$, it must be that at least one of the points $p_1,\dots,p_k$ is inside $S_1$. Without loss of generality, assume $p_1\in S_1$.

    Let $h'$ be a hyperplane that goes through $p_1$ and is parallel to $h$. Then there are less than $T$ points on one side of $h'$, and hence $p_1$ has Tukey depth $T < \gamma n$, a contradiction.
\end{proof}

Helly's theorem~\cite{Danzer1963HellysTA} implies that a point with a high Tukey depth always exists.
\begin{fact}[\cite{Danzer1963HellysTA}]~\label{fact:Tukey depth of median}
    For every $S \in (\bbR^d)^n$, the center point (and hence also the Tukey median) of $S$ has Tukey depth at least $\frac{n}{d+1}$ in $S$.
\end{fact}

In the following lemma, we apply Theorem~\ref{thm:alphaApprox} (the $\alpha$-approximation theory of Vapnik and Chervonenkis~\cite{VC}) to determine an upper bound on the subset size $m$ such that a given set of points $S$ can be $(\alpha,\beta,m)$-approximated with respect to the Tukey depth function (Definition~\ref{def:approximation-wrt-Q}).\footnote{Similar statements are also provided in some works about approximating center points (e.g.~\cite{CEMST93}).}

\begin{lemma}\label{lem:VC theory 1}

Let $S \in (\bbR^d)^n$ and let $S'\subseteq S$ be a random subset of $S$ with cardinality $m=|S'| > O\left(\frac{d\cdot\log(d/\alpha)+\log(1/\beta)}{\alpha^2}\right)$. 
Then, with probability at least $1-\beta$, for all $p\in \mathbb{R}^d$, if $TD_{S'}(p)=\gamma' m$ and $TD_{S}(p)=\gamma n$ then $|\gamma-\gamma'|\leq\alpha$.
\end{lemma}
\begin{proof}
    This proof uses the well-known result that $d$ dimensional half-space has VC-dimension $d$. More exactly, for $X = \bbR^d$ and the set of all halfspaces $R$, the VC dimension $(X, R)$ is $d$. So by Theorem~\ref{thm:alphaApprox}, with probability $1-\beta$, an $m$-size random subset $S'$ is an $\alpha$-approximation of $S$ with respect to $R$ (Definition~\ref{def:alpha-approx}), which implies that for any hyper-plane, if there are $\gamma_1 n$ points of $S$ and $\gamma_2 m$ points of $S'$ on one side of $h$, then it holds that $|\gamma_1-\gamma_2|\leq\alpha$. In the following, we assume that this $1-\beta$ probability event occurs (denote by $E$).

   Fix a point $p$ with $TD_{S}(p)=\gamma n$, and assume towards a contradiction that $TD_{S'}(p) = \gamma' m$ for $\gamma' < \gamma-\alpha$ (the direction $\gamma' > \gamma + \alpha$ follows similarly). This implies that there exists a hyperplane that goes through $p$ such that there are at most $\gamma' m$ points of $S'$ on one side of it. Since event $E$ occurs, the same side of the hyperplane contains at most $(\gamma' + \alpha) n$ points of $S$. The latter implies that $TD_{S}(p)\leq (\gamma' + \alpha) n < \gamma n$, contradiction. We thus conclude that $|\gamma-\gamma'|\leq\alpha$.

\end{proof}
\subsubsection{Domain Extension}

In our algorithm, we find an approximate Tukey median one coordinate at a time, and it is possible that a point would fall outside the input domain $\calX^d$. We use the domain extension technique by~\cite{beimel2019private}.
\begin{lemma}[domain extension of $\calX$~\cite{beimel2019private}]\label{lem:domainExtension}
    Given a finite $\calX\subset \mathbb{R}$, there exist sets $\tilde{\calX}_1,\dots,\tilde{\calX}_d$, where  $|\tilde{\calX}_i|\leq (d|\calX|^{d^2(d+1)})^{2^d}$, such that for all $1\leq i\leq d$, for all $S\subset \calX^d$, and for all $(x^*_1,\dots,x^*_{i-1})\in \tilde{\calX}_1\times\dots\times\tilde{\calX}_{i-1}$, there exist $(\tilde{x}_i,\dots,\tilde{x}_d)\in\tilde{\calX}_i\times\dots\times\tilde{\calX}_d$ satisfying
    $$
    \max_{x_i,\dots,x_d\in\mathbb{R}} TD_S(x^*_1,\dots,x^*_{i-1},x_i,\dots,x_d)=TD_S(x_1^*,\dots,x^*_{i-1},\tilde{x}_i,\dots,\tilde{x}_d).
    $$
\end{lemma}

\subsection{Privately Estimating the Tukey Median}~\label{sec:general-tukey-median}
We use $\IPConcaveHighDim$ (Algorithm~\ref{alg:IPConcaveHighDim}) to find a point with high Tukey depth privately. Define the function 
\begin{align*}
    Q_{TD}(S,x)=\frac{TD_S(x)}{|S|}.
\end{align*}
Here we verify that $Q_{TD}(S,x)$ satisfies the accuracy requirements in Theorem~\ref{thm:IPConcaveHighDim}.

\begin{enumerate}
    \item \textbf{Approximation by a random subset:} By Lemma~\ref{lem:VC theory 1}, a random subset $S'\subseteq S$ of size $$m(\alpha,\beta) \in O\left(\frac{d\cdot\log(d/\alpha)+\log(1/\beta)}{\alpha^2}\right)$$ is an $\alpha$-approximation of $S$ with respect to $Q_{TD}$ with probability $1-\beta$.

    \item \textbf{Concavity:} It is guaranteed by Fact~\ref{fact:tukey depth of point in convex hull}.

    \item \textbf{Proper finite domains:} The domain extension $\tcX_1,\ldots,\tcX_d$ in Lemma~\ref{lem:domainExtension} provide proper finite domains for $Q_{TD}$ (Definition~\ref{def:proper-finite-domains}) with maximal domain size $X$ where $\log^* X \in O(\log^* (\size{\cX} + d))$.
\end{enumerate}

Thus, the following theorem is an immediate corollary of Theorem~\ref{thm:IPConcaveHighDim}.

\begin{theorem}\label{thm:TD:1}
    Let $\eps > 0$, $\delta, \alpha,\beta \in (0,1)$, $n,d,t \in \bbN$ with $n \geq t$ and let $\cX \subseteq \bbR$ be a finite domain. Let $\Ac \colon (\cX^d)^n \rightarrow \bbR^d$ be the algorithm that on input $S$, computes $\IPConcaveHighDim_{\alpha,\beta,\eps,\delta,t}(S, Q_{TD})$ with the finite domains $\tcX_1, \tcX_2, \ldots, \tcX_d$ from Lemma~\ref{lem:domainExtension} and outputs its output, and let $X = \max_{i \in [d]}\set{|\tcX_i|}$. Then
    \begin{enumerate}
        \item \textbf{Privacy:} $\Ac$ is $(\eps \cdot\sqrt{2d\ln(1/\delta')},d\delta+\delta')$-differentially private for any choice of $\delta'>0$.

        \item \textbf{Accuracy:} Let $\lambda_{opt} \size{S}$ be the Tukey depth of the Tukey median in $S \in \cX^n$. Assuming that $t = n_{IP}(X,\beta,\eps,\delta)$ (the sample complexity of $\PrivateIP$ in Theorem~\ref{thm:1-d interior point}) and that
        \begin{align*}
            n \geq m(\alpha,\beta)\cdot n_{IP}(X,\beta,\eps,\delta)
            &\in O\paren{\frac{d \log(d/\alpha) + \log(1/\beta)}{\alpha^2} \cdot \frac{\log^* (\size{\cX} + d) \cdot \log^2\paren{\frac{\log^* (\size{\cX} + d)}{\beta \delta}}}{\eps}}\\
            &= \tilde{O}\paren{\frac{\log^* \size{\cX}\cdot (d + \log(1/\beta)) \cdot \log^2(1/\delta)}{\eps \alpha^2}},
        \end{align*}
        then with probability $1 - (t + d) \beta$, $\Ac(S)$ outputs a point $\hx \in \bbR^d$ with $TD_S(\hx) \geq (\lambda_{opt} - 2d\alpha) n$.
    \end{enumerate}
\end{theorem}

Recall that by Fact~\ref{fact:Tukey depth of median}, the Tukey median of any $S \in \cX^d$ has Tukey depth at least $\lambda_{opt} n$ for $\lambda_{opt} = \frac1{d+1}$. Thus, by substituting $\varepsilon$ by $\frac{\varepsilon}{\sqrt{2d\ln(2/\delta)}}$, $\delta$ by $\frac{\delta}{2d}$, $\delta'$ by $\delta/2$, $\alpha$ by $\frac{\alpha}{2d(d+1)}$, and $\beta$ by $\frac{\beta}{t + d}$, we obtain our main theorem for estimating the Tukey median.

\begin{theorem}[Restatement of Theorem~\ref{thm:intro:TD}]\label{thm:TD}
    Let $\eps > 0$, $\delta, \alpha,\beta \in (0,1)$, $n,d \in \bbN$ and let $\cX \subseteq \bbR$ be a finite domain. There exists an $(\eps,\delta)$-differentially private algorithm $\Ac \colon (\cX^d)^n \rightarrow \bbR^d$ and a value
    \begin{align*}
        n_{min} \in 
        \tilde{O}\paren{\frac{d^{4.5}\cdot \log^* \size{\cX} \cdot (d + \log(1/\beta))\cdot \log^{2.5}(1/\delta))}{\eps \alpha^2}}
        =\tilde{O}_{\alpha,\beta,\varepsilon,\delta}\left(d^{5.5}\cdot\log^*|\calX|\right),
    \end{align*}
    such that if $n \geq n_{min}$, then for any $S \in (\cX^d)^n$ it holds that
    \begin{align*}
        \ppr{\hx \sim \Ac(S)}{TD_S(\hx) \geq \frac{1 - \alpha}{d+1}\cdot n} \geq 1 - \beta.
    \end{align*}
\end{theorem}

\section{Privately Learning Halfspaces via Privately Solving Linear Feasibility Problems}\label{sec:halfspaces}

In this section, we show how to use $\IPConcaveHighDim$ (Algorithm~\ref{alg:IPConcaveHighDim}) to solve linear feasibility problems privately and thus imply a private half-spaces learner. We give additional preliminaries in Section~\ref{sec:pre-linear-feasbility}, and describe our algorithm for solving linear feasibility problems in Section~\ref{sec:alg-linear-feasibility}.

\subsection{Additional Preliminaries for Linear Feasibility Problems}~\label{sec:pre-linear-feasbility}

\paragraph{Notation.} For $a\in\mathbb{Z}^+$ we use $[[a]]$ to denote the set $\{-a,-a+1,\dots,a\}$. For two sets of integers $A,B$ and a scalar $c$, we define $A/B=\{a/b: a\in A \wedge b\in B\}$ and $c\cdot A = \{ca: a\in A\}$. For $x = (x_1,\ldots,x_d)$ and $y = (y_1,\ldots,y_d)$, we denote by $\ip{x,y} = \sum_{i=1}^d x_i y_i$ the inner-product of $x$ and $y$.

Recall that for $a\in\mathbb{R}^d$ and $w\in\mathbb{R}$, we define the predicate $h_{a,w}(x) = \left(\langle a,x\rangle \geq w\right)$. Let $H_{a,w}$ be the halfspace $\{x\in\mathbb{R}^d \colon h_{a,w}(x) = 1\}$.

\paragraph{Linear Feasibility Problem~\cite{kaplan2020private}.}
Let $X\in\mathbb{N}$ be a parameter and $\mathcal{X}=[[X]]$. 
In a linear feasibility problem, we are given a feasible collection $S\in(\mathcal{X}^{d+1})^n$ of $n$ linear constraints over $d$ variables $x_1,\dots,x_d$. The target is to find a solution in $\mathbb{R}^d$ that satisfies all (or most) constraints. Each constraint has the form $h_{a,w}(x) = 1$ for some $a\in\mathcal{X}^d$ and $w\in \mathcal{X}$. I.e., a linear feasibility problem is an LP problem with integer coefficients between $-X$ and $X$ and without an objective function.

\begin{definition}[$(\alpha,\beta)$-solving $(X,d,n)$-linear feasibility~\cite{kaplan2020private}]
    We say that an algorithm $(\alpha,\beta)$-solves $(X,d,n)$-linear feasibility if for every feasible collection of $n$ linear constraints over $d$ variables with coefficients in $\mathcal{X}$, with probability $1-\beta$ the algorithm finds a solution $(x_1,\dots,x_d)$ that satisfies at least $(1-\alpha)n$ constraints.
\end{definition}

Notice that there exists a reduction from PAC learning of halfspaces in $d$ dimensions when the domain of labeled examples is $\calX^d$ to solving linear feasibility. 
Each input labeled point $((x_1,\dots,x_d),y)\in\calX^d\times\{-1,1\}$ is transformed to a linear constraint $h_{y\cdot(x_1,\dots,x_d,-1),0}$.
By Theorem~\ref{thm:learn vc}, if we can $(\alpha,\beta)$-solve $(X,d+1,n)$ linear feasibility problems with $n\geq \tilde{O}_{\alpha,\beta}(d)$, then the reduction results in an $(O(\alpha),O(\beta))$-PAC learner for $d$ dimensional halfspaces.\footnote{We remark that the reduction only works smoothly when the points are assumed to be in general position, but this assumption can be eliminated (see Section 5.1 in \cite{kaplan2020private}).}

For any point, it is natural to consider how many linear constraints it satisfies. We define it as the $depth$ of this point. However, as mentioned by \cite{kaplan2020private}, the $depth$ function is not quasi-concave, i.e., for some points $p_1,\dots,p_k$, the $depth$ of the point inside the convex hull of $p_1,\dots,p_k$ cannot guarantee to be as high as $p_1,\dots,p_k$ (see the left figure of Figure~\ref{fig:depth-cdepth}, $A$ and 
$B$ have $depth$ 1, but $C$ only has $depth$ 0). Kaplan et. al.~\cite{kaplan2020private} consider the convexification of the $depth$, which is called $cdepth$. The property of $cdepth$ guarantees that the $cdepth$ of any point inside the convex hull of $p_1,\dots,p_k$ is as high as that of $p_1,\dots,p_k$ (see the right figure of Figure~\ref{fig:depth-cdepth}, $A$ and 
$B$ have $cdepth$ 1, $C$ also has $cdepth$ 1).
\begin{figure}[ht]
\includegraphics[scale=.5]{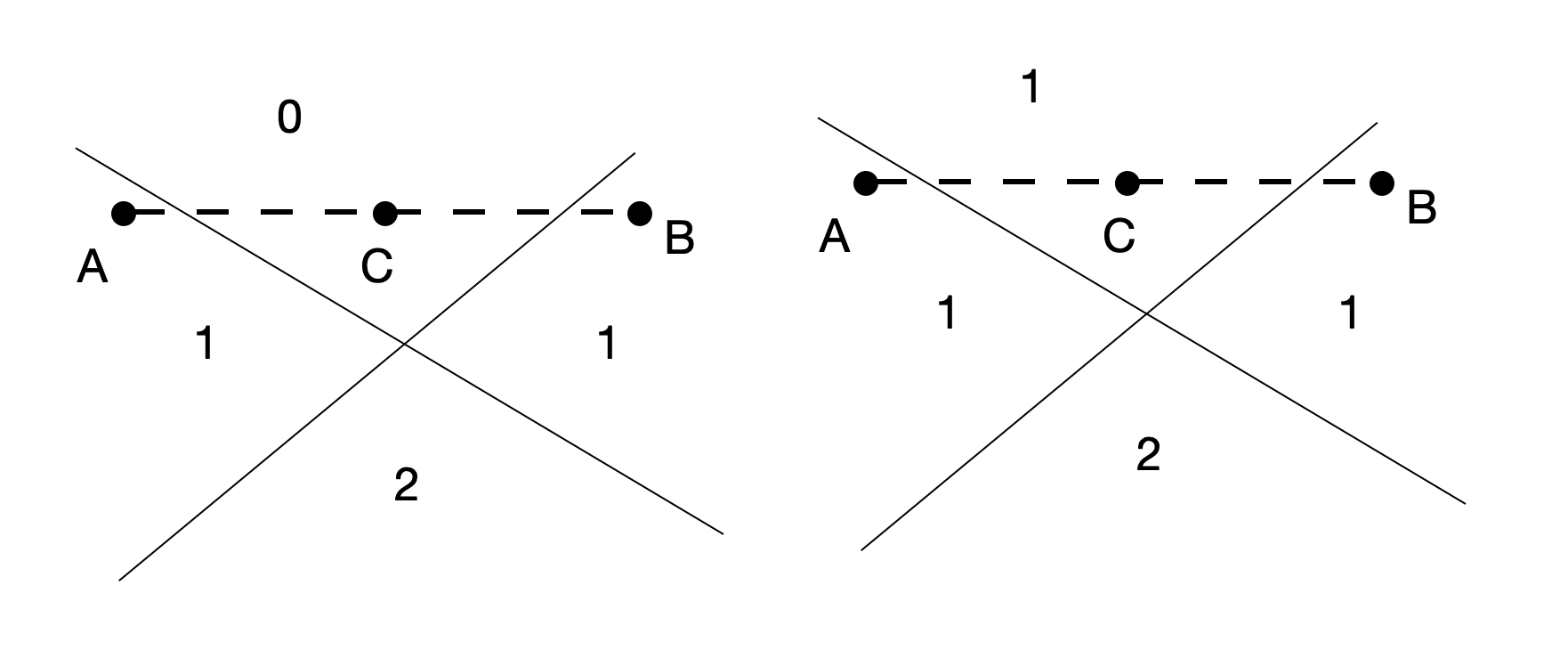}
\caption{Depth and cdepth \label{fig:depth-cdepth}}
\end{figure}

\begin{definition}[Convexification of a function $f:(\mathcal{X}^{d+1})^*\times\mathbb{R}^d\rightarrow \mathbb{R}$~\cite{kaplan2020private}]
    For $S\in(\mathcal{X}^{d+1})^*$ and $y\in\mathbb{R}$, let $\mathcal{D}_{S}(y)=\{z\in\mathbb{R}^d:f(S,z)\geq y\}$. The convexification of $f$ is the function $f_{Conv}:\mathcal{X}^*\times\mathbb{R}^d\rightarrow \mathbb{R}$ defined by $f_{Conv}(S,x)=\max\{y\in\mathbb{R}:x\in ConvexHull(\mathcal{D}_{S(y)})\}$.
\end{definition}
I.e., if $x$ is in the convex hall of points $Z\subset \mathbb{R}^d$\remove{such that $f(S,z)\geq f(S,x)$ for all $z\in Z$} then $f_{Conv}(S,x)$ is at least $\min_{z\in Z}(f(S,z))$. 

\begin{definition}[$depth$ and $cdepth$~\cite{kaplan2020private}]\label{def:depth-cdepth}
Let $S$ be a collection of predicates. Define $depth_{S}(x)$ to be the number of predicates $h_{a,w}$ in $S$ such that $h_{a,w}(x)=1$.
Let $cdepth_{S}(x)=f_{Conv}(\mathcal{S},x)$ for the function $f(S,x)=depth_{\mathcal{S}}(x)$.
\end{definition}

Similarly to Fact~\ref{fact:tukey depth of point in convex hull}, we have
\begin{fact}~\label{fact:convex of cdepth}
    Let $S$ be a set of $n$ linear constraints and $p_1,\dots,p_k$ be points with $cdepth_S(p_i)\geq \lambda n$ for all $i\in[k]$. Then any point $p$ in the convex hull of $\{p_i\}_{i\in[k]}$ satisfies $cdepth_S(p)\geq \lambda n$.
\end{fact}
\begin{proof}
    By the assumption, each of the points $p_i$ can be written as a convex combination $p_i=\sum_{j\in[k_i]} \eta_{i,j}y_{i,j}$ of  $k_i$ points $y_{i,1},\dots, y_{i,k_i}$ with $depth_S(y_{i,j})\geq\lambda n$. Since $p$ is in the convex hall of $p_1,\ldots,p_k$, it is also in the convex hull of $\{y_{i,j}\}$.
\end{proof}

\begin{fact}[\cite{kaplan2020private}]\label{fact:cdepth to depth}
    For any $S\in(\mathbb{R}^d\times \mathbb{R})^*$ and any $x\in\mathbb{R}^d$, it holds that 
    $$
    depth_{S}(x)\geq (d+1)\cdot cdepth_{S}(x)-d|S|.
    $$
\end{fact}

By the above fact, if we can find a point with $cdepth_{S}\geq (1-\alpha)|S|$ where $\alpha \ll 1/(d+1)$, then this point has $depth_S\approx |S|$.

Analogously to Lemma~\ref{lem:domainExtension}, we define the domain extension for linear feasibility problems. Unlike in Lemma~\ref{lem:domainExtension}, where the extension did not depend on the input for the problem of approximating the Tukey median, the extension for linear feasibility depends on the input. 
Given a collection $S$ of predicates, the extension is computed iteratively, coordinate by coordinate, where the input for the $i^{th}$ iteration is a prefix $(x^*_1,\ldots,x^*_{i-1})$ in the extension obtained in iterations $1$ to $i-1$. 
\begin{definition}[domain extension for linear feasibility~\cite{kaplan2020private}]\label{def:domain for linear feasibility}
    We define the $d$ domains $\{\tilde{\calX}_i\}^d_{i=1}$ iteratively. For $i=1$, let $\tilde{\calX}_1=\tilde{\calX}_1'/\tilde{\calX}_1''$ where $\tilde{\calX}_1':=[[(d\cdot d!)\cdot X^d]]$ and $\tilde{\calX}_1'':=([[d!\cdot X^d]])\backslash\{0\}$. 
    For $i>1$ and given $x^*_{i-1}=s_{i-1}/t_{i-1} \in \tilde{\calX}_{i-1}$ where $s_{i-1}\in \tilde{\calX}_{i-1}'$ and $t_{i-1}\in\tilde{\calX}_{i-1}''$, let $\tilde{\calX}_i=\tilde{\calX}_i'/\tilde{\calX}_i''$ where $\tilde{\calX}_i'=[[(d\cdot d!)^i\cdot X^{di}]]$ and $\tilde{\calX}_i''=([[ d!\cdot t_{i-1} \cdot X^d]])\backslash\{0\}$.
\end{definition}

\begin{theorem}[\cite{kaplan2020private}]\label{thm:domain for linear feasibility}
    Let $X\in\mathbb{N}$, $\calX\in[[\pm X]]$, $S\in(\calX^d\times \calX)^*$, $i\in[d]$.
    Assume that $\tilde{\calX}_j$ is according to Definition~\ref{def:domain for linear feasibility} for all  $j\in[i]$ where for $j<i$ the coordinates $x^*_j$ are picked from $\tilde{\calX}_j$. 
    Then there exists $x^*_i\in\tilde{\calX}_i$ such that
    $$   \max_{\tilde{x}_{i+1},\dots,\tilde{x}_d\in\mathbb{R}}cdepth_{S}(x_1^*,\dots,x_{i-1}^*,x_i^*,\tilde{x}_{i+1},\dots,\tilde{x}_d)=\max_{\tilde{x}_{i},\dots,\tilde{x}_d\in\mathbb{R}}cdepth_{S}(x_1^*,\dots,x_{i-1}^*,\tilde{x}_i,\tilde{x}_{i+1},\dots,\tilde{x}_d)
    $$
\end{theorem}

The following lemma gives the VC dimension of the linear feasibility problem so that we can apply the VC theory on it. The lemma is closely related to the fact that the VC dimension of $d$-dimensional halfspaces is $d$ \cite{VC,matousek2013lectures,Sauer,shelah1972combinatorial}, and we give the full proof details in Appendix~\ref{sec:vc linear feasibility}.

\def\VCLinearFeasLemma{
    Let $X_{Halfspace}=\{H_{a,w} \mid (a,w)\in \mathbb{R}^{d+1}\}$.
    For a point $p \in \bbR^d$, let $r_p = \set{H_{a,w} \mid p \in H_{a,w}}$, and let $R_{Points} = \set{r_p \mid p \in \bbR^d}$.
    The VC dimension of $(X_{Halfspace},R_{Points})$ is $d$.
}

\begin{lemma}\label{lem:vc linear feasibility}[VC dimension of linear feasibility]
    \VCLinearFeasLemma
\end{lemma}

Note that for a dataset $S = \set{(a_1,w_1),\ldots,(a_n,w_n)} \in (\bbR^{d+1})^n$ and a point $p \in \bbR^d$, the depth of $p$ in $S$ is $\size{\set{H_{a_1,w_1},\ldots,H_{a_n,w_n}} \cap r_p}$. Hence, the following statement is an immediate corollary of Theorem~\ref{thm:alphaApprox} and Lemma~\ref{lem:vc linear feasibility}.

\begin{corollary}
\label{cor:VC theory-linear feasibility}

Let $S\subseteq (\mathbb{R}^d)^n$ and let $S'\subseteq S$ be a random subset of $S$ with cardinality $m=|S'|\geq O\left(\frac{d\cdot\log(\frac{d}{\alpha})+\log\frac{1}{\beta}}{\alpha^2}\right)$. 
Then, with probability at least $1-\beta$, for all $p\in \mathbb{R}^d$, if $depth_{S'}(p)=\gamma' m$ and $depth_{S}(p)=\gamma n$ then $|\gamma-\gamma'|\leq\alpha$.
\end{corollary}

In the following, we show that the same approximation guarantee holds also with respect to the $cdepth$ function.

\begin{corollary}~\label{corollary:vc for cdepth}
    Let $S\subseteq (\mathbb{R}^d)^n$ and 
    let $S'\subseteq S$ be a random subset of $S$ with cardinality $m=|S'|\geq O\left(\frac{d\cdot\log(\frac{d}{\alpha})+\log\frac{1}{\beta}}{\alpha^2}\right)$. 
    Then, with probability at least $1-\beta$, for all $p\in \mathbb{R}^d$, if $cdepth_{S'}(p)=\gamma' m$ and $cdepth_{S}(p)=\gamma n$ then $|\gamma-\gamma'|\leq\alpha$.
\end{corollary}
\begin{proof}
    Assume $p$ is a point with $cdepth_S(p)=\gamma n$. By the definition of $cdepth$, there exists points $p_1,\dots,p_k$, such that $depth_{S}(p_i) \geq \gamma n$ for all $i\in[k]$ and $p\in ConvexHull(p_1,\dots,p_k)$. By Corollary~\ref{cor:VC theory-linear feasibility}, with probability $1-\beta$, $depth_{S'}(p_i)\geq (\gamma-\alpha) m$ for all $i\in[k]$. Therefore, $cdepth_{S'}(p)\geq (\gamma-\alpha) m$. The proof for the other direction is similar. 
\end{proof}

\subsection{Privately Solving Linear Feasibility Problems}~\label{sec:alg-linear-feasibility}
We use $\IPConcaveHighDim$ (Algorithm~\ref{alg:IPConcaveHighDim}) to solve linear feasibility problems privately. Define the function 
\begin{align*}
    Q_{LF}(S,x)=\frac{cdepth_S(x)}{|S|}.
\end{align*}
Here we verify that $Q_{LF}(S,x)$ satisfies the accuracy requirements in Theorem~\ref{thm:IPConcaveHighDim}.

\begin{enumerate}
    \item \textbf{Approximation by a random subset:} By Corollary~\ref{corollary:vc for cdepth}, a random subset $S'\subseteq S$ of size $$m(\alpha,\beta) \in O\left(\frac{d\cdot\log(d/\alpha)+\log(1/\beta)}{\alpha^2}\right)$$ is an $\alpha$-approximation of $S$ with respect to $Q_{LF}$ with probability $1-\beta$.

    \item \textbf{Concavity:} It is guaranteed by Fact~\ref{fact:convex of cdepth}.

    \item \textbf{Proper finite domains:} The domain extension $\tcX_1,\ldots,\tcX_d$ in Definition~\ref{def:domain for linear feasibility} and Theorem~\ref{thm:domain for linear feasibility} provide proper finite domains for $Q_{LF}$ (Definition~\ref{def:proper-finite-domains}) with maximal domain size $X$ where $\log^* X \in O(\log^* (\size{\cX} + d))$.
\end{enumerate}

Thus, the following theorem is an immediate corollary of Theorem~\ref{thm:IPConcaveHighDim}.

\begin{theorem}
    Let $\eps > 0$, $\delta, \alpha,\beta \in (0,1)$, $n,d,t \in \bbN$ with $n \geq t$ and let $\cX \subseteq \bbR$ be a finite domain. Let $\Ac \colon (\cX^d)^n \rightarrow \bbR^d$ be the algorithm that on input $S$, computes $\IPConcaveHighDim_{\alpha,\beta,\eps,\delta,t}(S, Q_{LF})$ with the finite domains $\tcX_1, \tcX_2, \ldots, \tcX_d$ from Definition~\ref{def:domain for linear feasibility} and outputs its output, and let $X = \max_{i \in [d]}\set{|\tcX_i|}$. Then
    \begin{enumerate}
        \item \textbf{Privacy:} $\Ac$ is $(\eps \cdot\sqrt{2d\ln(1/\delta')},d\delta+\delta')$-differentially private for any choice of $\delta'>0$.

        \item \textbf{Accuracy:}  Assuming that $t = n_{IP}(X,\beta,\eps,\delta)$ (the sample complexity of $\PrivateIP$ in Theorem~\ref{thm:1-d interior point}) and that
        \begin{align*}
            n \geq m(\alpha,\beta)\cdot n_{IP}(X,\beta,\eps,\delta)
            &\in O\paren{\frac{d \log(d/\alpha) + \log(1/\beta)}{\alpha^2} \cdot \frac{\log^* (\size{\cX} + d) \cdot \log^2\paren{\frac{\log^* (\size{\cX} + d)}{\beta \delta}}}{\eps}}.\\
            &= \tilde{O}\paren{\frac{\log^* \size{\cX}\cdot (d + \log(1/\beta)) \cdot \log^2(1/\beta\delta)}{\eps \alpha^2}},
        \end{align*}
        Then with probability $1 - (t + d) \beta$, $\Ac(S)$ outputs a point $\hx \in \bbR^d$ with $cdepth_S(\hx) \geq (1 - 2d\alpha) n$.
    \end{enumerate}
\end{theorem}

Combining the corollary with Fact~\ref{fact:cdepth to depth}, we have the following result.
\begin{corollary}
    With probability $1-(t+d)\beta$, the point $\hx$ satisfies $depth_{S}(\hx)\geq(1-2d\alpha-2d^2\alpha)\cdot |S|$.
\end{corollary}

Thus, by substituting $\varepsilon$ by $\frac{\varepsilon}{\sqrt{2d\ln(2/\delta)}}$, $\delta$ by $\frac{\delta}{2d}$, $\delta'$ by $\delta/2$, $\alpha$ by $\frac{\alpha}{4d^2}$, and $\beta$ by $\frac{\beta}{t + d}$, we obtain our main theorem for solving linear feasibility problem.

\begin{theorem}\label{thm:privateLinearFeasibility}
    Let $\eps > 0$, $\delta, \alpha,\beta \in (0,1)$, $n,d \in \bbN$ and let $\cX \subseteq \bbR$ be a finite domain. There exists an $(\eps,\delta)$-differentially private algorithm $\Ac \colon (\cX^d)^n \rightarrow \bbR^d$ and a value 
    \begin{align*}
        n_{min} \in 
        \tilde{O}\paren{\frac{d^{4.5}\cdot \log^* \size{\cX} \cdot (d + \log(1/\beta))\cdot \log^{2}(1/\beta\delta)\cdot\log^{0.5}(1/\delta)}{\eps \alpha^2}}
        =\tilde{O}_{\alpha,\beta,\varepsilon,\delta}\left(d^{5.5}\cdot\log^*|\calX|\right),
    \end{align*}
    such that if $n \geq n_{min}$, then for any $S \in (\cX^d)^n$ it holds that
    \begin{align*}
        \ppr{\hx \sim \Ac(S)}{depth_S(\hx) \geq (1 - \alpha)\cdot n} \geq 1 - \beta.
    \end{align*}
\end{theorem}

\subsection{Privately Learning Halfspaces}
Combining Theorems~\ref{thm:privateLinearFeasibility} and~\ref{thm:learn vc}, we get the following result.
\begin{theorem}[Restatement of Theorem~\ref{thm:intro:halfspaces}]
    Let $\cX \subset \bbR$ be a finite domain. 
    There exists an $(\eps,\delta)$-differentially private $(\alpha,\beta)$-PAC learner for halfspaces over examples from $\cX^d$ with sample complexity $n = \tilde{O}_{\alpha,\beta,\varepsilon,\delta}\left(d^{5.5}\cdot\log^*|\cX|\right)$.
\end{theorem}

\bibliographystyle{plainnat}

\appendix

\section{Approximated Functions Have Low Sensitivity}\label{sec:approx-is-low-sen}

In this section, we show that any approximated function has low sensitivity.

\begin{theorem}
    If $(S,Q)$ can be $(\alpha,1/3,m)$-approximated for any $S$ (Definition~\ref{def:approximation-wrt-Q}), then for any neighboring datasets $S_1,S_2$ with size at least $4m$, we have $|Q(S_1,x)-Q(S_2,x)|\leq 2\alpha$.
\end{theorem}
\begin{proof}
    Assume $|S_1|=|S_2|=n$ and $x_1$ and $x_2$ are the different entries in the two datasets.
    For each of $S_1$ and $S_2$, the number of $m$-size subsets 
    is $\binom{n}{m}$. 
    So for each of $S_1$ and $S_2$, the number of $\alpha$-approximation is at least $\frac{2\binom{n}{m}}{3}$. Since the number of $m$-size subsets containing $x_1$ or $x_2$ are at most $\binom{n-1}{m-1}$, thus the number of $\alpha$-approximation that does not 
    containing $x_1$ or $x_2$ is at least $\left(\frac{2\binom{n}{m}}{3}+\frac{2\binom{n}{m}}{3}\right)-\binom{n}{m}-\binom{n-1}{m-1}=\left(\frac{1}{3}-\frac{m}{n}\right)\cdot\binom{n}{m}>0$. Let $S'$ be one of these $\alpha$ approximations. Then $|Q(S_1,x)-Q(S_2,x)|\leq |Q(S_1,x)-Q(S',x)|+|Q(S',x)-Q(S_2,x)|\leq2\alpha$
\end{proof}
\remove{
We conclude it  Figure~\ref{fig:Relationship}:
\begin{figure}[h]
\includegraphics[scale=.6]{Relationship.png}
\caption{Sample complexity of private optimization for quasi-concave functions~\label{fig:Relationship}}
\end{figure}
}

\section{VC Dimension of Linear Feasibility Problem}\label{sec:vc linear feasibility}
In this section, we give a proof of Lemma~\ref{lem:vc linear feasibility}, restated below.

\begin{lemma}[Restatement of Lemma~\ref{lem:vc linear feasibility}]
    \VCLinearFeasLemma
\end{lemma}
The proof has two parts. Recall that for $a\in\mathbb{R}^d$ and $w\in\mathbb{R}$ we define the predicate $h_{a,w}(x) = \left(\langle a,x\rangle \geq w\right)$. Let $H_{a,w}$ be the halfspace $\{x\in\mathbb{R}^d \mid h_{a,w}(x) = 1\}$. 

\paragraph{1. There exists $d$ halfspaces that can be shattered.} Consider $d$ halfspaces $S=\{H_{e_1,0},\dots H_{e_d,0}\}$, where $e_i=0^{i-1}\times1\times0^{d-i}$. Note that for every $x \in \set{-1,1}^d$ and every $i \in [d]$: $H_{e_i,0} \in r_x \iff x_i = 1$, where recall that $r_x = \set{H_{a,w} \mid x \in H_{a,w}}$. Therefore, $S$ is shattered by $\set{r_x}_{x \in \set{-1,1}^d}$.

\paragraph{2. Every $d+1$ halfspaces cannot be shattered.} For any halfspace $H_{a,w}$, and any point $p$ on one side of the hyperplane $\{x \mid \ip{a,x}+w=0\}$, the function $r_p$ give the same label to $H_{a,w}$. Given $d+1$ halfspaces, consider the hyperplane arrangement of their corresponding hyperplanes. 
Each cell of the hyperplane arrangement can be uniquely projected to one dichotomy of $d+1$ halfspaces (see an example in Figure~\ref{fig:vc linear feasibility}). By Theorem~\ref{thm:number of cells in hyperplane arrangement}, the number of all possible dichotomies is at most $\sum_{i=0}^d\binom{d+1}{d}<2^{d+1}$. Thus $d+1$ halfspaces cannot be shattered.
\begin{figure}
\begin{center}
\includegraphics[scale=.5]{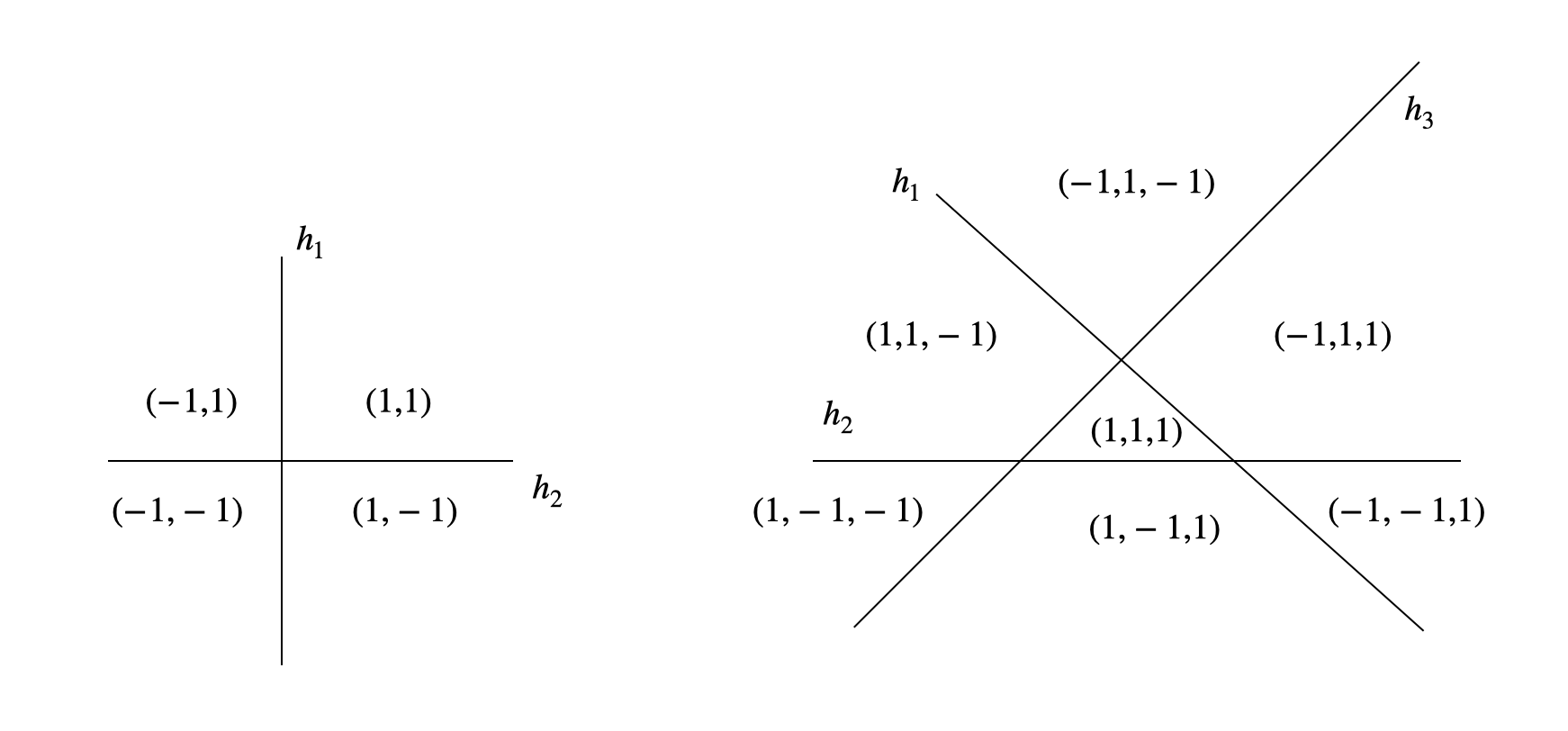}
\end{center}
\caption{2 halfspaces can be shattered (left) and 3 halfspaces cannot be shattered (right) in the 2-D plane \label{fig:vc linear feasibility}}
\end{figure}

\begin{theorem}[\cite{matousek2013lectures} Page 127, Proposition 6.1.1]\label{thm:number of cells in hyperplane arrangement}
    The maximum number of cells for $n$ hyperplanes with $n>d$ in $\mathbb{R}^d$ is $\sum_{i=0}^d\binom{n}{d}$.
\end{theorem}

\section{Huang et al. \cite{ACML-huang25b}'s Algorithm is Not Differentially Private}\label{sec:ACML}

Huang et al.~\cite{ACML-huang25b} present an algorithm that they claim is differentially private and learns halfspaces over examples from $\cX^d$ with sample complexity $\tilde{O}(d \cdot \log^* \size{\cX})$. In this section, we explain why their algorithm is inherently not differentially private. For simplicity, we even focus on their $2$-dimensional version (described in their Section 3.1).

Huang et al.~\cite{ACML-huang25b}'s 2-dimensional learning algorithm is focused on learning halfspaces that go through the origin. Each such halfspace can be uniquely represented by an angle $\phi \in [0,2\pi)$ (denote this halfspace by $h_{\phi}$), so their goal is to describe an algorithm $\Ac \colon (\cX^2 \times \set{-1,1})^* \rightarrow [0,2\pi)$ that is: (1) $(\eps,\delta)$-differentially private, and (2) Given a dataset $S \in (\cX^2 \times \set{-1,1})^n$ with $n \geq \tilde{\Theta}_{\alpha,\beta,\eps,\delta}(\log^* \size{\cX})$ that is realizable by some $h_{\phi}$ (unknown $\phi \in [0,2\pi)$), with probability $1-\beta$, $\Ac(S)$ outputs $\phi^*$ such that $\size{\set{(x,y) \in S \colon h_{\phi^*}(x) = y}} \geq (1-\alpha)n$ (i.e., $A$ is $(\alpha,\beta)$-empirical learner).

Huang et al.~\cite{ACML-huang25b} first define a finite discretization $\cH_{\gamma}$ of $[0,2\pi)$ with size $O(\size{\cX}^2)$ such that if $S \in \cX^d$ is realizable over $[0,2\pi)$, then it is also realizable over the grid $\cH_{\gamma}$.

\begin{algorithm}\label{alg:ACML:MakeData}
    \caption{MakeData}

    \textbf{Inputs:} $\eps > 0$, $\cH_{\gamma} \subseteq [0,2\pi)$, $S \in (\cX^2 \times \set{-1,1})^*$.

    \textbf{Operation:}~

    \begin{enumerate}
        \item $S_{\cH} \la \emptyset$.

        \item for $\phi \in \cH_{\gamma}$ do:

        \begin{enumerate}
            \item $n_{\phi} = \size{\set{(x,y) \in S \colon \size{\phi(x) - \phi} < \gamma \text{ and }h_{\phi}(x) = y}}$.

            \item Add $\max\set{\ceil{n_{\phi} + \Lap(1/\eps)},1}$ copies of $\phi$ to $S_{\cH}$.
        \end{enumerate}
        
        \item Return $S_{\cH}$.

    \end{enumerate}    
\end{algorithm}

\begin{algorithm}\label{alg:ACML:MakeThrData}
    \caption{MakeThrData}

    \textbf{Inputs:} $S_{\cH} \in ([0,2\pi])^*$, $S \in (\cX^2 \times \set{-1,1})^*$, $C \in \bbN$.

    \textbf{Operation:}~

    \begin{enumerate}
        \item Calculate $q(S,\phi) = \size{\set{(x,y) \in S \colon h_{\phi}(x) = y}}$ for every $\phi \in S_{\cH}$.
        \item Let $\max_C(S_{\cH})$ be the $C$ largest elements in $S_{\cH}$ according to the lexicographic order of $(q(S,\phi),\phi)$;
        \item Randomly select $\phi' \in S_{\cH}\setminus \max_C(S_{\cH})$ and rotate the coordinate so that $\phi' = 0$;
        \item Let $\phi^* \eqdef \argmax_{\phi \in \max_C(S_{\cH})}\set{q(S,\phi)}$.
        \item $S_{Thr} \la \emptyset$.
        \item For $\phi \in \max_C(S_{\cH})$ do
        \begin{enumerate}
            \item $y \la 1$ if $\phi \leq \phi^*$; otherwise, $y \la -1$.
            \item Add $(\phi,y)$ to $S_{Thr}$.
        \end{enumerate}
        return $S_{Thr}$.
    \end{enumerate}
    
\end{algorithm}

\begin{algorithm}\label{alg:A_{SimpleH}}
    \caption{$A_{SimpleH}$}
    \textbf{Inputs:} $\eps, \delta, \alpha, \beta > 0$, $S \in (\cX^2 \times \set{-1,1})^*$, $\gamma \in [0,2\pi)$, and $(\eps,\delta)$-differentially private $(\alpha,\beta)$-empirical learner $A_{Thr}$ that privately learn thresholds over $\cX_{Thr}$ with sample complexity $n_{Thr} = n_{Thr}(\cX_{Thr},\eps,\delta,\alpha,\beta)$.

    \textbf{Operation:}~

    \begin{enumerate}
        \item $\cH_{\gamma} \la Discretize(\gamma)$;
        \item $S_{\cH} \la MakeData(\eps,\cH_{\gamma},S)$;
        \item $S_{Thr} \la MakeThrData(S_{\cH}, S, n_{Thr})$;
        \item Apply $A_{Thr}$ with input $S_{Thr}$, parameters $\eps,\delta,\alpha,\beta$, and get $\phi^*$.
        \item Output $\phi^*$.
    \end{enumerate}
\end{algorithm}

\begin{theorem}[Theorem 14 in \cite{ACML-huang25b}]\label{thm:ACMLpaper}
    For any $\eps,\delta,\alpha,\beta \in (0,1)$, if there is an $(\eps,\delta)$-differentially private $(\alpha,\beta)$-empirical learner $A_{Thr}$ that learns thresholds on a finite domain $\cX_{Thr}$ with $n_{Thr}(\cX_{Thr},\eps,\delta,\alpha,\beta)$ samples, then with sample complexity 
    \begin{align*}
        n = O\paren{n_{Thr}(\cX_{Thr}, \eps/2, \delta,\alpha,\beta)},
    \end{align*}
    $A_{SimpleH}$ (Algorithm~\ref{alg:A_{SimpleH}}) is an $(\eps,\delta)$-differentially private $(\alpha,\beta)$-empirical learner for $2$-dimensional halfspaces.
\end{theorem}

In the following, we prove that Theorem~\ref{thm:ACMLpaper} is wrong by showing that Algorithm~\ref{alg:A_{SimpleH}} is blatantly not differentially private.

\paragraph{Counterexample to Theorem~\ref{thm:ACMLpaper}.}

Let $S$ be a dataset that consists of $n/2+1$ copies of $((1,0),-1)$ (i.e., the point $(1,0)$ with a label $-1$), and $n/2-1$ copies of $((-1,0),-1)$, and let $S'$ be the neighboring dataset that is obtained by replacing one of the $((1,0),-1)$ in $S$ with $((-1,0),-1)$. We remark that although $S$ and $S'$ are not realizable (i.e., no halfspace agree on the labeled points), differential privacy is a \emph{worst-case} guarantee and so $A_{SimpleH}(S)$ (with privacy parameters $\eps,\delta$) should be $(\eps,\delta)$-indistinguishable from $A_{SimpleH}(S')$ by Theorem~\ref{thm:ACMLpaper}.\footnote{We remark that it is also not hard to determine realizable datasets $S,S'$ that break the privacy guarantee of $A_{SimpleH}$, but we chose the unrealizable ones as they make the arguments cleaner and simpler.}

Let $S_{\cH}$, $S_{Thr}$, and $\phi^*$ be the values computed in the execution $A_{SimpleH}(S)$, and let $S_{\cH}'$, $S_{Thr}'$, and $\phi'^*$ be corresponding values in the execution of $A_{SimpleH}(S')$. Note that both $S_{\cH}$ and $S_{\cH}'$ contain at least $1$ copy of every angle in the grid $\cH_{\gamma}$. Furthermore, note that in algorithm $MakeThrData$, for any grid angle $\phi \in (0,\pi)$ we have $q(S,\phi) = n/2 + 1$ and $q(S',\phi) = n/2 - 1$, and for any grid angle $\phi \in (\pi,2\pi)$ we have $q(S,\phi) = n/2 - 1$ and $q(S',\phi) = n/2 + 1$. Therefore, the output $S_{Thr}$ of $MakeThrData(S_{\cH},S)$ is a dataset of angles in $(0,\pi)$, and the output $S_{Thr}'$ of $MakeThrData(S_{\cH}',S')$ is a dataset of angles in $(\pi,2\pi)$. Thus, the output $\phi^*$ of $A_{Thr}(S_{\cH})$ is in $(0,\pi)$, and the output $\phi'^{*}$ of $A_{Thr}(S_{\cH}')$ is in $(\pi,2\pi)$, and thus we conclude that $A_{SimpleH}(S)$ and $A_{SimpleH}(S')$ are clearly distinguishable.

\paragraph{What is wrong in \cite{ACML-huang25b}'s Privacy Analysis?}

The main issue in \cite{ACML-huang25b}'s privacy analysis is the (wrong) argument that the composition of $MakeThrData$ on top of $MakeData$ is ``differentially private'' just because $MakeData$ is differentially private, which is wrong since $MakeThrData$ also uses the input dataset. As we demonstrated in our counter example, this composition is clearly not differentially private, and in fact it can result in completely disjoint outputs on neighboring datasets.

\end{document}